\documentclass[aip,
author-year,
]{revtex4-1}
\usepackage{amsmath}    % need for subequations
\usepackage{amssymb}
\usepackage{amsthm}
\usepackage{graphicx}   % need for figures
\usepackage{verbatim}   % useful for program listings
\usepackage{color}      % use if color is used in text
\usepackage{subfigure}  % use for side-by-side figures
\usepackage{hyperref}   % use for hypertext links, including those to external documents and URLs

\usepackage{xfrac}
\usepackage{natbib}

\usepackage[T1]{fontenc}

\usepackage{tikz}
\usepackage{url}

\usepackage{color}

\usepackage[mathscr]{eucal}

\usepackage{graphicx}
\graphicspath{pic/}
\newcommand{\Reals}{{\mathbb R}}         %
\newcommand{\Co}{{\mathbb C}}         %
\newcommand{\half}{\frac{1}{2}}         %

\newcommand{\Arg}{\operatorname{Arg}}

\newcommand{\KDelta}{\Delta}

\newcommand{\rs}{r_*}

\newcommand{\RadialOp}{\mathbf R}
\newcommand{\AngularOp}{\mathbf S}

\newcommand{\Rhor}{R_{\text{hor}}}
\newcommand{\Rout}{R_{\text{out}}}
\newcommand{\rsharp}{r_\sharp}

\newcommand{\TMESolu}{\Phi}

\newcommand{\Cont}{\mathcal{C}}
\newcommand{\Hankel}{\mathcal{H}}

\newcommand{\xiplus}{\xi}
\newcommand{\ximinus}{\eta}

\newcommand{\Rcal}{\mathcal{R}}
\newcommand{\Tcal}{\mathcal{T}}

\newcommand{\deltaplus}{\delta}
\newcommand{\deltaminus}{\gamma}

\newcommand{\sfrak}{\mathfrak{s}}

\newtheorem{thm}{Theorem} 
\newtheorem{cor}[thm]{Corollary} 
\newtheorem{lemma}[thm]{Lemma} 
 
\newtheorem{definition}[thm]{Definition} 

\newtheorem{remark}[thm]{Remark} 

\begin{document}

\title{Mode stability on the real axis}

\author{Lars Andersson}
\email{lars.andersson@aei.mpg.de}
\author{Siyuan Ma}
\email{siyuan.ma@aei.mpg.de}
\author{Claudio Paganini}
\email{claudio.paganini@aei.mpg.de}
\affiliation{Albert Einstein Institute, Am M\"uhlenberg 1, D-14476 Potsdam,  Germany }
\author{Bernard F. Whiting}
\email{bernard@phys.ufl.edu}
\affiliation{Department of Physics, University of Florida, 2001 Museum Road, Gainesville, FL 32611-8440, USA}


%\date{\today \ {\em File:\jobname{.tex}}}
\date{9 September 2016} 

\begin{abstract}
A generalization of the mode stability result of \cite{whiting:1989} for the Teukolsky equation is proved for the case of real frequencies. The main result of the paper states that a separated solution of the Teukolsky equation governing massless test fields on the Kerr spacetime, which is purely outgoing at infinity, and purely ingoing at the horizon, must vanish. This has the consequence, that for real frequencies, there are linearly independent fundamental solutions of the radial Teukolsky equation $\Rhor, \Rout$, which are purely ingoing at the horizon, and purely outgoing at infinity, respectively. This fact yields a representation formula for solutions of the inhomogenous Teukolsky equation.
\end{abstract}

\maketitle

\tableofcontents

\section{Introduction}

The field equations on the Kerr spacetime for massless test fields with spins
$\sfrak$ between 0 and 2 
imply that the scalar components with extreme spin weights $s = \pm \sfrak$ solve the Teukolsky Master Equation (TME) \citep{teukolsky:1973}, a separable, spin-weighted wave equation.
Let
\begin{align*}
{\mathbf L} &= \partial_r \KDelta \partial_r
- \frac{1}{\KDelta} \left \{ (r^2+a^2) \partial_t + a \partial_\phi - (r-M)s
\right \}^2 \\
&\quad - 4s (r+ia\cos\theta)\partial_t + \frac{1}{\sin\theta} \partial_\theta
\sin\theta \partial_\theta
\\
&\quad
+ \frac{1}{\sin^2\theta} \left \{ a \sin^2\theta \partial_t + \partial_\phi +
is \cos\theta \right \}^2
\end{align*}
where $(t,r,\theta,\phi)$ are Boyer-Lindquist coordinates and $\Delta = r^2 -2Mr+a^2$. Then \citep{whiting:1989}
\begin{equation}\label{eq:TME}
{\mathbf L} \TMESolu = 0
\end{equation}
is a form of the TME on the Kerr exterior background with mass $M$ and angular momentum per unit mass $a$. The parameter $s$ is the spin weight of the field $\TMESolu$.

For completeness, we recall the defininition of the fields $\TMESolu$ solving the TME. In order to do this, we make the spin weight explicit as a subindex. For $s=0$, the TME is equivalent to the scalar wave equation $\nabla^a \nabla_a \TMESolu_0 = 0$. For 
spins $\sfrak = 1/2, 1, 3/2, 2$ the field equations are Dirac-Weyl, Maxwell, Rarita-Schwinger, and linearized gravity, respectively. For spins $0, 1/2, 1, 2$, see \cite{teukolsky:1973}, for the spin-3/2 case, \cite{1992PhRvD..46.5395T}, see also \cite{1995JMP....36.6929S}.
Working in the Kinnersley principal tetrad, let $\phi_0$, $\phi_2$ be the Newman-Penrose scalars of spin weights $1, -1$ for a Maxwell test field on the Kerr background, and let  $\dot \Psi_0, \dot \Psi_4$ denote the linearized Weyl scalars of spin weights $2, -2$ for a solution of the linearized vacuum Einstein equations on the Kerr background, see \cite{2011CQGra..28f5001A} for details. Let the scalar $\zeta$ be chosen so that $\zeta \propto \Psi_2^{-1/3}$, where $\Psi_2$ is the spin weight zero Weyl scalar. In Boyer-Lindquist coordinates, we can take $\zeta = r-ia\cos\theta$. The scalar fields $\TMESolu_s$ for integer $s$ are defined by setting
\begin{subequations}\label{eq:TMESolu-def}
\begin{alignat}{2}
\TMESolu_{-2} &= \Delta^{-1} \zeta^4 \dot \Psi_4, &\quad \TMESolu_{-1} &= \Delta^{-1/2} \zeta^2 \phi_2, \\
\TMESolu_1     &= \Delta^{1/2} \phi_0, &\quad \TMESolu_2 &= \Delta \dot \Psi_0
\end{alignat}
Similarly, let $\chi_0, \chi_1$ denote the scalars of spin weights $\pm 1/2$ for a Dirac-Weyl test field, and $H_0, H_3$ the scalars of spin weights $\pm 3/2$ for a Rarita-Schwinger test field. We define
\begin{alignat}{2}
\TMESolu_{-3/2} &= \Delta^{-3/4} \zeta^3 H_3, &\quad \TMESolu_{-1/2} &= \Delta^{-1/4} \zeta \chi_1, \\
\TMESolu_{1/2}     &= \Delta^{1/4} \chi_0, &\quad \TMESolu_{3/2} &= \Delta^{3/4} H_0
\end{alignat}
\end{subequations}

The TME admits separated solutions of the form
\begin{equation}\label{eq:Psi-sep}
\TMESolu = e^{-i\omega t} e^{im\phi} S(\theta)R(r).
\end{equation}
where $\omega, m$ are the frequencies corresponding to the Killing vector fields $(\partial_t)^a$, $(\partial_\phi)^a$. Let
\begin{equation}\label{eq:Kdef}
K = (r^2+a^2)\omega - am
\end{equation}
Then with
\begin{align}
\RadialOp ={}& \partial_r \KDelta \partial_r + \frac{K^2 - 2i K(r-M) s - (r-M)^2 s^2 }{\KDelta} + 4sir\omega - \Lambda \label{eq:RadialOp}\\
\AngularOp ={}&
\frac{1}{\sin\theta} \partial_\theta \sin\theta \partial_\theta  - \frac{m^2}{\sin^2\theta} + a^2 \cos\theta^2 \omega^2 - 2 a \omega s  \cos\theta - \frac{2ms\cos\theta }{\sin^2\theta} - s^2 \cot^2 \theta  \nonumber \\
&{}+ \Lambda + 2a \omega m - a^2 \omega^2 ,\label{eq:AngularOp}
\end{align}
where $\Lambda$ is a separation constant, which can be assumed to be real for real $\omega$, we have after making the substitutions $\partial_t \leftrightarrow -i\omega$, $\partial_\phi \leftrightarrow im$,
\begin{align*}
\mathbf{L} ={}& \RadialOp + \AngularOp,
\intertext{and}
[\RadialOp, \AngularOp] ={}& 0 .
\end{align*}
In particular, $\RadialOp, \AngularOp$ are commuting symmetry operators for $\mathbf{L}$. It follows from the above that for separated waves of the form \eqref{eq:Psi-sep}, \eqref{eq:TME} is equivalent to the equations $\RadialOp R = 0$, $\AngularOp S = 0$.
We shall refer to the equations
\begin{subequations}\label{eq:radang}
\begin{align}\label{eq:radialeq}
\RadialOp R ={}& 0 \\
\AngularOp S ={}& 0 \label{eq:angulareq}
\end{align}
\end{subequations}
as the radial and angular Teukolsky equations, respectively. We shall not be concerned with the analysis of the angular Teukolsky equation here, but point out that $\AngularOp$ is formally self-adjoint on $[0,\pi]$ with respect to $\sin\theta d\theta$. Requiring that the solutions correspond to regular spin-weighted functions fixes the boundary conditions at $\theta=0,\pi$ and equation \eqref{eq:angulareq} becomes a Sturm-Liouville problem which has a discrete, real spectrum\footnote{The separation constant used here is related to that used in \cite{1974ApJ...193..443T} by $\Lambda + 2a\omega m - a^2 \omega^2 = E-s^2$, and to the one used in \cite{whiting:1989} and \cite{1972PhRvL..29.1114T} by $\Lambda + 2a\omega m - a^2 \omega^2 = A+s$.}; see \cite{1986JMP....27.1238L} for more details.

For fields of non-zero spin, the TME does not admit a real action, and hence standard arguments do not yield energy conservation and dispersive estimates. This is an obstacle to proving stability for the test fields with non-zero spin on the Kerr exterior spacetime, which would be an important step towards proving non-linear stabilty of the Kerr black hole,  i.e. that a Kerr black hole with $|a| < M$ is dynamically stable as a solution to the vacuum Einstein field equations, in the sense that the maximal development of a sufficiently small perturbation of the Kerr solution is asymptotic in the future to a member of the Kerr family.

In \cite{whiting:1989}, one of the authors gave a proof of mode stability. In particular, the TME has no separated wave solutions, or modes, which are such that the frequency $\omega$ has positive imaginary part, and which have no incoming radiation in the sense that the wave is outgoing at infinity, and ingoing at the horizon, see section \ref{sec:noincoming}. The main result of \cite{whiting:1989} states that the TME admits no exponentially growing solutions without incoming radiation. In the case of $\Im \omega > 0$, the condition of no incoming radiation can be restated as saying that the solution has support only on the future horizon and null infinity. 
On the other hand, there do exist mode solutions with no incoming radiation for certain frequencies with negative imaginary part. This case corresponds to quasi-normal modes \citep{1999LRR.....2....2K}, which are exponentially decaying in time.

It is known that exponentially growing modes must arise by quasi-normal frequencies passing from the lower half plane through the real axis into the upper half plane as $a$ is changed from zero. This was argued heuristically by 
Press and Teukolsky 
\citeyearpar[p. 651]{1973ApJ...185..649P} and later shown by 
Hartle and Wilkins 
\citeyearpar{1974CMaPh..38...47H}, see also \citet[p. 452]{1974ApJ...193..443T}. For this reason, the mode stability problem can be reduced to considering the case of real frequencies.

Recently, the mode stability argument has been revisited for the case of real frequencies, restricting to the spin-0 case \citep{2015AnHP...16..289S}.  In the case of real frequencies, the mode stability result states that restricting to separated waves with no incoming radiation in the above sense, the radial Teukolsky equation has no non-trivial solutions. This has the consequence that there are linearly independent solutions $\Rhor, \Rout$ which are purely ingoing at the horizon and outgoing at infinity, respectively, a fact which plays a central role in the proof of boundedness and decay for scalar waves on sub-extreme Kerr exterior spacetimes with $|a| < M$  \citep{2014arXiv1402.7034D}, in particular it is used to treat the superradiant range of frequencies; see section \ref{sec:superradiance} for background on superradiance. Some aspects of the argument in \cite{2015AnHP...16..289S} are discussed in Remark \ref{rem:yakov} below; see also the recent paper  \citep{2016arXiv160608005F}  which contains results implying mode stability on the real axis, proven using methods which are quite different from those used here.

Motivated by the relevance of the TME for the black hole stability problem we here give a proof of mode stability on the real axis for fields with arbitrary spin. Our main result is the following, cf. Theorem \ref{thm:realmodestab} below.

\begin{thm}[Mode stability on the real axis] \label{thm:main-intro} Let $\TMESolu$ be a separated solution to the TME with $\omega > 0$ for the sub-extreme Kerr black hole. Assume that $\TMESolu$ has purely ingoing radiation at the horizon and purely outgoing radiation at infinity. Then $\TMESolu = 0$.
\end{thm}

\begin{remark}
A classical scattering argument can be used to show mode stability on the real axis for  half-integer spins, or for frequences outside of the superradiant range, see equation \eqref{eq:superrad-cond} below. 
The proof of mode stability on the real axis presented in this paper is independent of that scattering argument. 

\end{remark}

The fact that there are no solutions to the TME with no incoming radiation has the important consequence that the radial Teukolsky equation has two fundamental solutions $R_{\text{hor}}$ and $R_{\text{out}}$ which are ingoing at the horizon, and outgoing at infinity, respectively, and are linearly independent, with non-vanishing Wronskian. This implies that one can construct solutions of the inhomogenous Teukolsky equation using the method of variation of the parameter, see Remark \ref{rem:inhomog} below for details. The properties of the solutions $R_{\text{hor}}$ and $R_{\text{out}}$ can be used to estimate the solution of the inhomogenous Teukolsky equation.

Unless otherwise stated, we shall in the rest of the paper restrict to positive frequency, $\omega > 0$. Note that the substitution $(\omega,m,s) \to (-\omega,-m,-s)$ maps solutions of TME to solutions.

In case $\omega=0$, \eqref{eq:Psi-sep} represents a time independent solution of the TME. In this case, the radial Teukolsky equation becomes a hypergeometric equation with three regular singular points $r_-, r_+, \infty$ which requires a separate discussion. It has been pointed out that this equation does not have solutions which are well-behaved at the horizon and at infinity, see \citet[p. 1117]{1972PhRvL..29.1114T}, see also \citet[p. 651]{1973ApJ...185..649P}.

\subsection*{Overview of this paper} In section \ref{sec:radialteuk} we collect some background on the radial Teukolsky equation and discuss the asymptotic behavior of its solutions in section \ref{sec:asymptotics}. Lemma \ref{lem:R-asympt} collects the facts about solutions with no incoming radiation which we shall need for the proof of our main result. The phenomenon of superradiance is reviewed in section \ref{sec:superradiance}. This analysis yields the previously known fact that for non-superradiant frequencies and for half-integer spins, mode stability holds. Section \ref{sec:inttrans} introduces the integral transformation which will be used along the lines of \cite{whiting:1989} to transform the radial Teukolsky equation to a non-superradiant equation. This transformation is the essential step in the proof of mode stability. The limiting behavior of the transformed radial function is analyzed in section \ref{sec:limits}, and the proof of the main theorem is given in section \ref{sec:modestab}.

\section{The radial Teukolsky equation} \label{sec:radialteuk}
The radial Teukolsky equation is a second order ordinary differential equation with rational coefficients, and an analysis of its singular points yields that it is a confluent Heun equation \citep{slavyonov:lay,ronveaux:heun},  with regular singular points at $r_-, r_+$ and an irregular singular point of rank\footnote{Here we use the notion of rank following \citet[p. 60]{MR0078494}. The s-rank of $r=\infty$ as in \citep{slavyonov:lay} is $2$.} $1$ at $\infty$. In this context it is natural to consider the radial Teukolsky equation on the complex $r$-plane. In a neighborhood of each regular singular point, the general solution is a linear combination of the fundamental Frobenius solutions, while at the irregular singular point, one considers the Thom\'e, or normal, solutions.  These are formal solutions but in each Stokes sector, a pair of solutions can be found which are asymptotically represented by the normal solutions, cf. \citet[Chapter III]{MR0078494} for details, see also \cite{MR1429619}.

Equation \eqref{eq:radialeq} with $\RadialOp$ given by \eqref{eq:RadialOp} is in self-adjoint form. 
The form of the equation and its solutions can be changed by transformations of the independent variable $r$ (eg. M\"obius transformations),  s-homotopic transformations (rescalings),  and integral transformations of the dependent variable. In the rest of this section, as a preparation for the integral transformations considered in section \ref{sec:inttrans}, we will state a few basic facts about the radial Teukolsky equation and its solutions.

The rotation speeds $\omega_{\pm}$ and surface gravities $\kappa_{\pm}$ of the outer and inner horizons located at $r_{\pm} = M \pm \sqrt{M^2 - a^2}$
are given by
\begin{align}
\omega_\pm ={}& a/(2Mr_{\pm}) , \label{eq:omegapm} \\
\kappa_{\pm} ={}& \pm \frac{r_+ - r_-}{4Mr_\pm} \label{eq:kappapm}
\end{align}
Associated to null generators of the horizons,
\begin{equation}\label{eq:chi-def}
\chi_\pm^a = (\partial_t)^a + \omega_\pm (\partial_\phi)^a,
\end{equation}
we have for a wave with time frequency $\omega$ and azimuthal frequency $m$, the frequencies
\begin{align} \label{eq:kpm-def}
k_{\pm} ={}& \omega - \omega_\pm m .
\end{align}
The superradiant range of frequencies is characterized, independent of the signs of $\omega$, $m$ and $a$, by 
\begin{equation} \label{eq:superrad-cond}
\omega k_+ < 0
\end{equation}
The tortoise coordinate $\rs$ is defined by
\begin{align}
d\rs ={}& \frac{r^2+a^2}{\KDelta} dr
\end{align}
We have
\begin{align}
\rs \sim{}&  \frac{1}{2\kappa_+} \ln(r-r_+), \quad \text{ for $r \to r_+$} \\
\rs \sim{}& r + 2M\ln(r), \quad \text{ for $r \to \infty$}
\end{align}
Let $\rsharp$ be defined by
$$
d\rsharp = \frac{a}{\Delta} dr
$$
and set $u_\pm = t\pm r_*$ and $\phi_{\pm} = \phi \pm \rsharp$. Then $(u_+, r, \theta, \phi_+)$ and $(u_-,r, \theta,\phi_-)$ are the ingoing and outgoing Eddington-Finkelstein coordinates, respectively. The ingoing coordinates are regular on the future horizon, while the outgoing ones are regular on the past horizon.

For later use, it will be convenient to introduce the quantities
\begin{subequations}\label{eq:xietadef}
\begin{align}
\xiplus ={}& \frac{i (am - 2Mr_+ \omega)}{r_+ - r_-} = -  i \frac{k_+}{2\kappa_+} \label{eq:xidef} \\
\ximinus ={} & \frac{-i(am-2Mr_-\omega)}{r_+-r_-} = - i \frac{k_-}{2\kappa_-} \label{eq:etadef}
\end{align}
\end{subequations}

\subsection{Asymptotics} \label{sec:asymptotics}
We shall now discuss the asymptotic behavior of the solutions of the radial Teukolsky equation \eqref{eq:radialeq} at the singular points. The radial function $R$ used here corresponds to the function $R$ used in \citep{teukolsky:1973} and \citep{1974ApJ...193..443T}, multiplied by a factor $\Delta^{s/2}$.
The asymptotic behavior of $R$ which we state below can be read off from  \cite[Table 1]{1974ApJ...193..443T} after taking the factor $\Delta^{s/2}$ into account, see also \cite{2013CQGra..30p5005C} for discussion of asymptotics. For completeness we shall indicate the derivation of these results.

\subsubsection{Asymptotics at $r=\infty$} \label{sec:asympt-infty}
In order to analyze the possible asymptotic behavior of solutions to the radial Teukolsky equation at infinity, it is convenient to transform the equation to normal form.
Write \eqref{eq:radialeq} as
\begin{equation}\label{eq:self-adj-R}
\partial_r \Delta \partial_r R + V_0 R = 0
\end{equation}
and let $G=\Delta^{-1/2}$. The rescaling $R=yG$ transforms \eqref{eq:self-adj-R} to the form
\begin{equation}\label{eq:y-normal}
\partial_r^2 y+ Q y = 0
\end{equation}
with
$$
Q = -\half \partial_r \frac{\partial_r \Delta}{\Delta} - \tfrac{1}{4} \left ( \frac{\partial_r \Delta}{\Delta} \right)^2 + \frac{V_0}{\Delta}
$$
The leading terms in $Q$ at $r=\infty$ are
$$
Q = \omega^2 +\frac{2i\omega s + 4 \omega^2 M}{r} + O(r^{-2}) .
$$

From this we can determine following \citet[Chapter III]{MR0078494} that the two normal solutions to the radial Teukolsky equation \eqref{eq:radialeq}
near $r=\infty$
have the asymptotic forms
\begin{equation}\label{eq:R-asympt}
R\sim e^{\pm i\omega r} r^{\pm 2iM\omega} r^{\mp s -1}
\end{equation}
with the upper sign corresponding to outgoing waves. Due to the fact that the singular point at $r=\infty$ is irregular, of rank $1$, we shall need some further details concerning the Stokes phenomenon. As mentioned above, the rank of the irregular singular point $r=\infty$ is $1$. Following \citet[Chapter III.5]{MR0078494} we find that the Stokes line is the real line $\Im r = 0$ in the complex $r$-plane. The Stokes line decomposes the complex $r$-plane into two Stokes sectors, and in each Stokes sector, one of the two normal solutions is exponentially increasing, and one is exponentially decreasing. These are referred to as the \emph{dominant} and \emph{recessive} solutions.

In particular, if we consider a sectorial region $S$, $|r| > A$, $\alpha < \Arg r < \beta$, contained in a Stokes sector, the asymptotic expansions of the two normal solutions hold uniformly in $\Arg r$, for $r \in S$. For later use, we note here that for a sectorial region $S$ in the upper half plane, with $0 < \alpha < \pi/2 < \beta < \pi$ as in figure \ref{fig:sector}, the outgoing condition at $r=\infty$, i.e. the upper sign choice in \eqref{eq:R-asympt}, corresponds to the recessive normal solution in $S$.

\begin{figure}[!h]

\begin{tikzpicture}[scale=0.4]

\begin{scope}[shading=radial,draw=white]
\draw[clip,rotate=30]
(1,0) -- (6,0)
arc [start angle=0, end angle=120, radius=6] -- (-0.5,0.866) arc [ start angle=120, end angle=0,radius=1] -- cycle;

\shadedraw[shading=radial,draw=white,opacity=10] (0,0) circle (6cm);
\end{scope}

\draw[very thin, blue, ->] (-6,0) -- (6,0) node[anchor=west]{$r$};
\draw[very thin,blue, ->] (0,-1) -- (0,7);

\draw (3,4) node{$S$};

\end{tikzpicture}
\caption{ } \label{fig:sector}
\end{figure}

\subsubsection{Asymptotics at $r_{\pm}$} \label{sec:asympt-rpm}
The characteristic exponents for regular singular points $r_+, r_-$ of the radial Teukolsky equation are
\begin{subequations}\label{eq:charexp-R}
\begin{align}
\{\xi-s/2, -\xi+s/2\} &\quad \text{for $r_+$} \\
\{\eta-s/2, -\eta+s/2\} &\quad \text{for $r_-$}
\end{align}
\end{subequations}
Let $r_0$ be one of the regular singular points, with characteristic exponents $\rho_j$, $j=1,2$, which we order so that $\Re \rho_1 \leq \Re \rho_2$. We first consider the non-resonant case where $\rho_1 - \rho_2$ is not an integer.   
In this case the Frobenius solutions at $r_0$ are of the form
\begin{equation}\label{eq:g-at-singular-point}
R_j(r) = (r-r_0)^{\rho_j} R_{0,j}(r), \quad j=1,2,
\end{equation}
where $R_{0,j}(r)$ are analytic in a neighborhood of $r_0$ of radius $r_+ - r_-$, cf. \citet[\S 1.1.4]{slavyonov:lay}, see also \citet[p. 60]{MR0078494}. We may normalize the solutions so that $R_{0,j}(r_+) = 1$.

We next consider the case of resonance at $r_+$, i.e. $\xiplus=0$ and $s$ an integer. Note that $\xiplus=0$ corresponds to the upper bound of the range of superradiant frequencies, see section \ref{sec:superradiance}. In this case, the characteristic exponents $\rho_j$ take values 
$-s/2, s/2$, and the Frobenius solutions of \eqref{eq:Tg=0} corresponding to $\rho_1$ contains a logarithmic term,
\begin{equation}\label{eq:R-resonant} 
R_1(r) = (r-r_+)^{\rho_1}R_{0,1} + A_2 (r-r_+)^{\rho_2} R_{0,2}(r) \ln(r-r_+) .
\end{equation}
Here, for$ R_{0,1}(R+)=1$, $A_2$ is a non-zero constant
that can be computed. In the resonant case, we choose $R_{0,1}$ so that $R_{0,1} = \sum_{n=0}^\infty c_n (r-r_+)^n$ with $c_{|s|} = 0$, cf. \citet[Theorem 1.3]{slavyonov:lay}.  
The case of resonance at $r_-$ is similar.

\subsubsection{Waves with no incoming radiation} \label{sec:noincoming}
We shall say that waves which are outgoing at infinity and ingoing at the horizon, i.e. $r_+$, satisfy the no incoming radiation condition. A discussion of the boundary conditions for the radial Teukolsky equation can be found in \cite[\S V]{teukolsky:1973}, where also the notion of ingoing and outgoing waves is defined, see also \cite[p. 653, eq. (2.10)]{1973ApJ...185..649P}.

As discussed by 
Penrose 
\citeyearpar{1965RSPSA.284..159P}, an analysis of the asymptotic behavior of massless fields at null infinity leads, upon taking into account the scaling properties of the tetrad components of the field, to specific rates of fall-off depending on the spin weight of the field. This is known as the peeling property; see \cite{2012JGP....62..867M, 1996CQGra..13..461F, 2014CMaPh.331..755A, 2011PhRvD..84b4036H} for discussions of various aspects of peeling.
The peeling property can be summarized by saying that for a scalar component $\varphi_s$ of spin weight $s$, defined with respect to the Kinnersley tetrad, we have $\varphi_s = O(r^{-3-s})$. Taking into account the rescalings given in \eqref{eq:TMESolu-def} we find for that the peeling behavior of the solution of the TME is
\begin{equation}\label{eq:TMESolu-peel}
\TMESolu_s = O(r^{-1-s}), \quad \text{as $r \to \infty$}.
\end{equation}

In order to analyze the behavior of spinning fields at the horizon, a tetrad which is well behaved at the horizon must be used. Following  Teukolsky and Press 
\citeyearpar[\S IV]{1974ApJ...193..443T}, see also 
Hawking and Hartle 
\citeyearpar{1972CMaPh..27..283H} and 
Znajek 
\citeyearpar{1977MNRAS.179..457Z}, one finds that the fields
$\Delta^{s/2} \TMESolu_s$ are regular on the horizon.

Outgoing waves at $r=\infty$ must have an asymptotic form compatible with peeling, as discussed above, and should have positive radial group velocity, while ingoing waves at the horizon as seen by a physically well-behaved observed must be \emph{non-special} (i.e. neither vanishing nor singular on the horizon) and should have negative radial group velocity. 

Based on the above discussion of Frobenius and Thom\'e solutions, we are led to the following definition. 
\begin{definition}[No incoming radiation condition] \label{def:noincoming} Let $R$ be a solution of the radial Teukolsky equation. 
Then we shall say that $R$ has no incoming radiation provided 
\begin{align}
R(r) \sim{}& e^{i\omega r} r^{2iM\omega} r^{-s-1} \quad \text{ at $r=\infty$,}  \label{eq:limr=infty}\\
R(r) \sim{}& (r-r_+)^{\xi-s/2} \quad \text{ at $r=r_+$}. \label{eq:limr=r+}
\end{align}
In particular, we shall require that $R(r)$ is equal to the Frobenius solution with exponent $\xi-s/2$ at $r=r_+$ and equal near infinity $r=\infty$ to the normal solution which is recessive in the upper half plane.
\end{definition}

Specializing the discussion in this section  to the case of waves with no incoming radiation, we can state the following lemma, which summarizes the properties that we shall make use of. Note that we here view $R$ as a solution of \eqref{eq:radialeq} in the complex $r$-plane. The results stated in the lemma are direct consequences of the discussion in this section and the references given there.
\begin{lemma} \label{lem:R-asympt}
Let $R$ be a solution to the radial Teukolsky equation with no incoming radiation. Then 
the following hold:
\begin{enumerate}
\item 
\label{point:noresonance} If $\xi\ne 0$ or if $s$ is not a positive integer, then near $r_+$,
\begin{alignat}{1} \label{eq:R-analyt}
R ={}& (r-r_+)^{\xiplus-s/2} R_{+,1}(r), %
\end{alignat}
where 
$R_{+,1}$ has a power series expansion in $r-r_+$ which converges in the disk $|r-r_+| < r_+ - r_-$.

\item \label{point:resonance} If $\xi=0$ and $s$ is a positive integer, then near $r_+$,
\begin{equation}\label{eq:R-log}
R = (r-r_+)^{-s/2} R_{+,1}(r) + A (r-r_+)^{s/2} \ln(r-r_+) R_{+,2}(r) \end{equation}
where $R_{+,1}, R_{+,2}$ have power series expansions in $r-r_+$ which converge in the disk $|r-r_+| < r_+ - r_-$. Here $A$ is a  constant which can be calculated from $R_{+,1}(r_+)$.
\item
Let $S$ be a sectorial region in the upper half $r$-plane, of the form $|r| > r_0$, $\alpha < \Arg r < \beta$ with $0 < \alpha < \beta < \pi$. Then $R$ has an asymptotic expansion
\begin{equation}\label{eq:R-factor}
R \sim \  e^{i\omega r} r^{2iM\omega} r^{-s-1} \sum_{n=0}^\infty c_n r^{-n},  \quad \text{as $r \to \infty$},
\end{equation}
which is valid uniformly in $S$. In particular, the estimate
\begin{equation}\label{eq:R-ordo}
e^{-i\omega r} R = O(r^{-s-1})
\end{equation}
is valid in $S$.
\end{enumerate}
\end{lemma}

\begin{remark} \label{rem:R-properties}
\begin{enumerate}
\item \label{point:Frob=0} 
It follows from the properties of the Frobenius solutions, cf. \citet[Theorem 1.3]{slavyonov:lay}, that in the non-resonant case, if
$R_{+,1}(r_+) =0$ then $R \equiv 0$. To see this, the coefficients in the expansion of $R_{+,1}$ are determined by $R_{+,1}(r_+)$, and hence if $R_{+,1}(r_+)=0$, then $R$ vanishes in a neighborhood of $r_+$ and hence must be identically zero. 
In the resonant case with $s=0$ and $\xi=0$, the logarithmic solution is excluded by condition \eqref{eq:limr=r+}. 
Finally, in the resonant case with $s>0$ and $\xiplus=0$, if $R_{+,1}(r_+) = 0$, we find that also $A=0$ and hence it follows that $R(r) \equiv 0$, see also the discussion in section \ref{sec:asympt-rpm}. 
\item \label{point:sector} 
The estimate in \eqref{eq:R-ordo} can be rephrased as saying that there is a constant $C$ such that in the sectorial region $S$,
$$
|e^{-i\omega r} R| \leq C|r^{-s-1}|
$$
The constant $C$ depends on the parameters of the system, and the limit
$$
\lim_{r\to +\infty} |r^{s+1} R(r)|
$$
where the limit is taken along the positive real line. By \cite[Chapter 7, Theorem 2.2]{MR1429619}, the asymptotic form \eqref{eq:R-factor} of $R$ is valid in a circular sector $|\Arg(-iz)| < 3\pi/2 - \delta$ for $\delta > 0$, and this is the maximal sector of validity.  
\item
For completeness, we record that the asymptotic representation along the real line can be stated in terms of the tortoise coordinate as
\begin{equation}\label{eq:Rout-asympt}
R \sim \left \{ \begin{array}{llc} e^{i\omega r} r^{2iM\omega} r^{-s-1} &{}\sim e^{i\omega \rs} r^{-s-1} & \text{as $r \to \infty$} \\ (r-r_+)^{\xiplus - s/2} &{}\sim \KDelta^{-s/2}e^{-ik_+ r_*} & \text{as $r \to r_+$} \end{array} \right .
\end{equation}
which, as mentioned above, after taking into account the rescaling by $\Delta^{s/2}$, agrees with the asymptotic form stated in \cite[Table 1]{1974ApJ...193..443T}.
\end{enumerate}
\end{remark}

\subsection{Superradiance} \label{sec:superradiance}
In this subsection we shall review the classical scattering analysis for spinning fields following Teukolsky and Press \citeyearpar{1974ApJ...193..443T}, see also \cite{MR1647491}. The results that we present here are not new, however, to the best of our knowledge the fact that superradiance does not happen for the spin-$\tfrac{3}{2}$ case has not been discussed before, see Remark \ref{rem:spin-3/2} below. We make the dependence of the spin weight explicit by a subindex $s$. Let $R_s$ be a solution of the radial Teukolsky equation \eqref{eq:radialeq} with spin weight $s$. The rescaling $\upsilon_s = (r^2+a^2)^{1/2} R_s$ transforms the radial Teukolsky equation to an equation with independent variable $r_* \in (-\infty,\infty)$, of the form
\begin{equation}\label{eq:Schrod}
\frac{d^2}{d\rs^2} \upsilon_s + V_s \upsilon_s = 0
\end{equation}
with
$$
V_s =  \frac{\KDelta}{(r^2+a^2)^2} V_{0,s} - V_{\text{resc}}
$$
where
\begin{equation}\label{eq:Vresc}
V_{\text{resc}} = \frac{1}{(r^2+a^2)} \frac{d^2}{d\rs^2} (r^2+a^2)^{1/2} = \frac{\Delta}{(r^2+a^2)^4} \left( a^2\Delta + 2Mr(r^2-a^2) \right ) .
\end{equation}
We shall refer to \eqref{eq:Schrod} as the Schr\"odinger form of the radial Teukolsky equation.
We have
\begin{subequations}\label{eq:Vasympt}
\begin{align}
\lim_{r\to r_+} V_s ={}&  (i \kappa_+ s +  k_+)^2, \\
\lim_{r \to \infty} V_s ={}& \omega^2 .
\end{align}
\end{subequations}
In particular, for $s\ne 0$, the potential $V_s$ is complex and
\begin{equation}\label{eq:Vconj}
\overline{V}_s = V_{-s},
\end{equation}
where $\overline{V}_s$ denotes the complex conjugate. Let $\upsilon_s, \upsilon_{-s}$ be (a priori independent) solutions of \eqref{eq:Schrod} with spin weight $s, -s$ respectively, and define the Wronskian
$$
W[\upsilon_s, \overline{\upsilon}_{-s}] = \upsilon_s' \overline{\upsilon}_{-s} -  \overline{\upsilon}_{-s}'\upsilon_s
$$
where we have used a ${\,}^\prime$ to denote $d/dr_*$. Due to \eqref{eq:Vconj}, both $\upsilon_s$ and $\overline{\upsilon}_{-s}$ solve the same equation, and hence the Wronskian is conserved,
\begin{equation}\label{eq:Wcons}
W[\upsilon_s, \overline{\upsilon}_{-s}]' = 0 .
\end{equation}
We now make a scattering ansatz for $ R_s(r)$ which is purely ingoing at the horizon and a superposition of an ingoing and an outgoing part at infinity,
\begin{equation}
        R_s(r) \sim \begin{cases} Y_{\text{hole},s}\Delta ^{-s/2} e^{-ik_+\rs} ,  & \text{ at } r=r_+\\
         Y_{\text{in}, s} e^{-i\omega \rs} r^{s-1}+ Y_{\text{out}, s} e^{i\omega \rs}r^{-s -1},  & \text{ at } r=\infty .
    \end{cases}
\end{equation}
Here we have intentionally left the normalization of the ingoing mode free. The fact that the Wronskian is conserved gives the identity
\begin{equation}\label{eq:evaluatedwronskian}
       - 4M r_+ (i k_+ +\sfrak\kappa_+) Y_{\text{hole},\sfrak}\overline{Y_{\text{hole},-\sfrak}}= -2i \omega Y_{\text{in},\sfrak}\overline{Y_{\text{in},-\sfrak}}+2i \omega Y_{\text{out},\sfrak}\overline{Y_{\text{out},-\sfrak}} .
\end{equation}

Following Teukolsky and Press \citeyearpar{1974ApJ...193..443T}, see also Starobinsky and Churilov \citeyearpar{1974JETP...38....1S}, we now use the Teukolsky-Starobinsky Identities (TSI) to establish a relation between the fields with spin weights $s, -s$. Define the operators $\mathcal{D}, \mathcal{D}^\dagger$ by $\mathcal{D}= (d/dr-iK/\Delta)$ and $\mathcal{D}^\dagger= \mathcal{D}(-\omega,-m)$. For real $\omega$ the dagger operation is identical with a complex conjugation. The TSI for the solutions of the radial Teukolsky equation are \cite{MR1647491}, see also \cite{1989JMP....30.2925K,2009PhRvD..80l4001F},
\begin{subequations}\label{eq:TSrelation}
\begin{align}
\Delta^{\sfrak/2} \mathcal{D}^{2\sfrak}\Delta^{\sfrak/2} R_{-\sfrak}={}& C_\sfrak R_\sfrak \label{eq:TSIs} \\
\Delta^{\sfrak/2} {\mathcal{D}^\dagger}^{2\sfrak}\Delta^{\sfrak/2} R_\sfrak ={}& \overline{C_\sfrak}R_{-\sfrak} \label{eq:TSI-s}
\end{align}
\end{subequations}
where $C_\sfrak$ are constants depending on the parameters $\sfrak,a,M,m,\omega$, and the separation constant $\Lambda$.
\begin{remark}
The TSI can be understood as saying that the operator $\Delta^{\sfrak/2} \mathcal{D}^{2\sfrak} \Delta^{\sfrak/2}$ applied to $R_{-\sfrak}$ is proportional to a solution $R_\sfrak$ of the radial Teukolsky equation with spin weight $\sfrak$ and vice versa. In applying the TSI we thus restrict to solutions $R_s, R_{-s}$ satisfying this condition. It can be shown, cf. \cite{2016arXiv160106084A} and references therein, that if the spin-weighted fields $R_s, R_{-s}$ are the radial functions corresponding to the components with extreme spin weight of a field satisfying the spin-$\sfrak$ test field equations on the Kerr background, then the TSI hold.
\end{remark}
Applying \eqref{eq:TSrelation} to the asymptotic solutions of $R_{\pm \sfrak}$ at the horizon and at infinity, and comparing leading order terms gives relations between $Y_{\text{hole},\pm \sfrak}$, $Y_{\text{in},\pm \sfrak}$, and $Y_{\text{out},\pm \sfrak}$. Some calculations give the identity
\begin{equation}\label{eq:scattering}
A_s |Y_{\text{hole},s}|^2 = |Y_{\text{in},s}|^2 - B_s |Y_{\text{out},s}|^2
\end{equation}
where
\begin{equation}\label{eq:Bs}
\begin{array}{ll}
B_\sfrak = \frac{(2\omega)^{4\sfrak}}{|C_\sfrak|^2},& B_{-\sfrak}=B_\sfrak^{-1}
\end{array}
\end{equation}
and $A_\sfrak$ are given by\footnote{For the current considerations this has been checked up to $\sfrak=2$ but it can be expected that the relation holds for all half-integer spins $\sfrak \geq0$.}
\begin{equation}\label{eq:Avals}
\begin{aligned}
A_{-\sfrak}&=\begin{cases}
\left(\frac{2Mr_+}{\omega}\right)^{2\sfrak+1} \frac{\prod_{n=0}^{ \sfrak }(k_+^2 + (\sfrak-n)^2\kappa_+^2)}{k_+} & \text{ if } \sfrak =0,1,2... \\
\left(\frac{2Mr_+}{\omega}\right)^{2\sfrak+1} \prod_{n=0}^{\lfloor \sfrak \rfloor}(k_+^2 + (\sfrak-n)^2\kappa_+^2)& \text{ if } \sfrak = \frac{1}{2},\frac{3}{2}...
\end{cases}
\\
A_{\sfrak}&= \left(\frac{2Mr_+}{\omega}\right)^2\frac{(k_+^2+\sfrak^2\kappa_+^2)} {A_{-\sfrak}}
\end{aligned}
\end{equation}
where $\lfloor \sfrak \rfloor$ is the integer part of $s$, i.e. the largest integer less than or equal to $\sfrak$. 
The reflection and transmission coefficients $\Rcal_\sfrak, \Tcal_\sfrak$ are defined as
$$
\Rcal_\sfrak = B_\sfrak \frac{|Y_{\text{out},\sfrak}|^2}{|Y_{\text{in},\sfrak}|^2}= B_{-\sfrak} \frac{|Y_{\text{out},-\sfrak}|^2}{|Y_{\text{in},-\sfrak}|^2}, \quad \Tcal_\sfrak =  A_\sfrak \frac{|Y_{\text{hole},\sfrak}|^2}{|Y_{\text{in},\sfrak}|^2}=A_{-\sfrak} \frac{|Y_{\text{hole},-\sfrak}|^2}{|Y_{\text{in},-\sfrak}|^2}
$$
Then $\Rcal_\sfrak$ represents the fraction of the ingoing wave energy 
which is reflected out to infinity, while $\Tcal_\sfrak$ represents the fraction which is transmitted into the black hole. By construction we have the conservation law $\Rcal_\sfrak + \Tcal_\sfrak = 1$.

From the values of the coeffients $A_s$ given in \eqref{eq:Avals} we see that for integer spins, $A_s$ changes sign with $\omega k_+$, while for half-integer spins, $A_s$ is positive. This means that when $0 < \omega/m < \omega_+$,  the reflection coefficient for a field with integer spin will be greater than unity. This phenomenon is known as superradiance.

If $A_s > 0$, i.e. for the non-superradiant frequencies or for the half-integer spins, then $\Rcal_s < 1$. In particular, this implies that if $Y_{\text{in},s} =0$, then also $Y_{\text{out},s} = 0$
so that a solution with no incoming radiation must be zero. Thus, in these cases mode stability holds for real frequencies, and the classical scattering analysis given here is sufficient to prove Theorem \ref{thm:main-intro}.

\begin{remark}\label{rem:spin-3/2} The fact that the coefficient $A_{\pm \half}$ is positive also for superradiant frequencies is known and follows from the fact that the spin-$\half$ field admits a future directed conserved current; see \cite{1998CRASM.327..743M} concerning the spin-$3/2$ case. 
\end{remark}

It remains to determine the square modulus of the constants $C_\sfrak$ in \eqref{eq:TSrelation}.  For this purpose we apply $\Delta^{\sfrak/2} {\mathcal{D}^\dagger}^{2\sfrak}\Delta^{\sfrak/2}$ to \eqref{eq:TSIs} and use \eqref{eq:TSI-s} to get
\begin{equation}\label{eq:staroconst}
      \Delta^{\sfrak/2}{\mathcal{D}^\dagger}^{2\sfrak}\Delta^{\sfrak} \mathcal{D}^{2\sfrak}\Delta^{\sfrak/2} R_{-\sfrak}= |C_\sfrak|^2R_{-\sfrak}
\end{equation}
The left hand side can now be evaluated using the TME for $R_{-\sfrak}$.

One 
finds\footnote{\label{foot:agree}For $\sfrak=1$ and $\sfrak=2$ this agrees with the expressions obtained in \citep{1974ApJ...193..443T} and \citep{MR1647491} after taking into account the different conventions used there.}
\begin{subequations}\label{eq:Csfrak}
\begin{align}
|C_{\half}|^2=& \Lambda + \frac{1}{2}\\
|C_1|^2=& (\Lambda + 1)^2 + 4 a m \omega -4 a^2 \omega^2\\
|C_{\tfrac{3}{2}}|^2=& (\Lambda + \frac{3}{2})^3+ (\Lambda + \frac{3}{2})^2+ 16(\Lambda + \frac{3}{2})(a^2 \omega^2-am\omega)-16 a^2 \omega^2\\
|C_2|^2=& (\Lambda + 2)^4+ 4 (\Lambda + 2)^3+ 4 (\Lambda + 2)^2(1+10 a m \omega -10 a^2 \omega ^2)\\
\nonumber & + 48(\Lambda + 2)(am\omega+ a^2 \omega^2) + 144 \omega^2(a^2 m^2 + M^2-2 a^3 m \omega + a^4 \omega^2)
\end{align}
\end{subequations}
\section{Integral transformations} \label{sec:inttrans}
It is convenient to transform the radial Teukolsky equation to its canonical form before introducing the integral transform that shall be used.
The %
rescaling
\begin{equation}\label{eq:R-to-g}
R(r) = (r-r_-)^{\ximinus-s/2} (r-r_+)^{\xiplus - s/2} e^{-i\omega r} g(r)
\end{equation}
puts the radial Teukolsky equation in canonical form.
Letting
\begin{subequations}\label{eq:heun-param}
\begin{align}
\deltaminus={}&
2\ximinus -s + 1\\
\deltaplus={}&
2\xiplus -s+1 \\
p={}&
-2i\omega \\
\alpha={}& 1-2s \\
\sigma={}& -\Lambda -2ir_-(1-2s)\omega-s
\end{align}
\end{subequations}
we have that $\RadialOp R = 0$ is equivalent to
\begin{equation}\label{eq:Tg=0}
T_r g(r) = 0
\end{equation}
where
\begin{multline}\label{eq:Toperator-def}
T_r h(r) = (r-r_-)(r-r_+) \frac{d^2 h}{dr^2} \\
+ (\deltaminus(r-r_+)+\deltaplus(r-r_-)+p(r-r_-)(r-r_+))\frac{dh}{dr} +(\alpha p(r-r_-)+\sigma) h
\end{multline}
is a Heun operator in canonical form, with parameters $\deltaminus,\deltaplus,p,\alpha,\sigma$.

Let $\tilde T$ be a new Heun operator with different parameters given by
\begin{subequations}\label{eq:tilde-heun-param}
\begin{align}
\tilde\deltaminus :={}& \alpha =  1-2s  \\
\tilde \deltaplus :={}&\deltaminus +\deltaplus -\alpha = 1-4iM\omega \\
\tilde p :={}& p \\
\tilde \alpha :={}& \deltaminus = 1 - s + 2\ximinus \\
\tilde \sigma :={}& \sigma
\end{align}
\end{subequations}
and let $f(x,r)$ be defined as
\begin{equation}\label{eq:f-def}
f(x,r) = e^{-p\frac{(x-r_-)(r-r_-)}{r_+-r_-}} .
\end{equation}
With the above choice of parameters for $\tilde T$ we have that
\begin{equation}
    (\tilde{T}_x-T_r) f(x,r)=0.
\end{equation}
As we shall see, this means that we can use $f(x,r)$ as the kernel for an integral transformation between solutions of these two Heun equations. Let a contour $\Cont$ in the complex $r$-plane be given and let $g(r)$ be a solution to \eqref{eq:Toperator-def} with parameters as in \eqref{eq:heun-param}.
Defining $\tilde g(x)$, following \cite{whiting:1989}, by
\begin{equation}\label{eq:gdefine}
\tilde{g}(x) = \int_\Cont f(x,r) (r-r_{-})^{\gamma-1}(r-r_{+})^{\delta -1}e^{pr}g(r)dr ,
\end{equation}
we have that
\begin{equation}
\begin{split}
    \tilde{T}_x \tilde{g}(x) &= \int_\Cont \tilde{T}_xf(x,r) (r-r_{-})^{\gamma -1}(r-r_{+})^{\delta -1}e^{pr}g(r)dr\\
    &=\int_\Cont T_r f(x,r) (r-r_{-})^{\gamma -1}(r-r_{+})^{\delta -1}e^{pr}g(r)dr\\
    &=
\left.
    (r-r_-)^\deltaminus (r-r_+)^\deltaplus e^{pr} \left( \frac{df(x,r)}{dr} g(r) - f(x,r) \frac{dg(r)}{dr} \right) 
\right|_{\partial \Cont}    \\
&+\int_\Cont f(x,r)(r-r_-)^{\gamma-1} (r-r_{+})^{\delta-1}e^{pr}T_rg(r)dr .
\end{split}
\end{equation}
 The last step is an integration by parts. Note that the expression in the last line vanishes identically, because $g(r)$ satisfies \eqref{eq:Toperator-def}. Hence, provided the integral in \eqref{eq:gdefine} converges and the boundary condition
 \begin{equation}\label{eq:bc}
\left. 
(r-r_-)^\deltaminus (r-r_+)^\deltaplus e^{pr} \left( \frac{df(x,r)}{dr} g(r) - f(x,r) \frac{dg(r)}{dr} \right) 
\right|_{\partial \Cont}
=0
\end{equation}
is satisfied, we see that $\tilde g$ satisfies the transformed equation
\begin{equation}\label{eq:transformedHeun}
\tilde T_x \tilde g(x) = 0
\end{equation}
Using the parameters \eqref{eq:heun-param} in equation \eqref{eq:gdefine} and using the relation \eqref{eq:R-to-g} we can write $\tilde g$ in the form
\begin{subequations}\label{eq:inttrans}
\begin{align}
\tilde g(x) ={}&  \int_{\Cont} e^{2i\omega\frac{(x-r_-)(r-r_-)}{r_+-r_-}} (r-r_-)^{2\ximinus-s}(r-r_+)^{2\xiplus-s} e^{-2i\omega r} g(r) dr \label{eq:int-g} \\
={}&  \int_{\Cont} e^{2i\omega\frac{(x-r_-)(r-r_-)}{r_+-r_-}}  (r-r_-)^{\ximinus-s/2}(r-r_+)^{\xiplus-s/2}e^{-i\omega r} R(r) dr \label{eq:int-R}
\end{align}
\end{subequations}
\begin{remark} Assuming no incoming radiation for $R$, we have
\begin{equation}\label{eq:int-asymp}
(r-r_-)^{\ximinus-s/2}(r-r_+)^{\xiplus - s/2} e^{-i\omega r} R(r) \sim \left\{
\begin{array}{cc} (r-r_+)^{2\xiplus -s} & \text{ for $r \to r_+$} \\
r^{-2s-1} & \text{ for $r \to \infty$}
\end{array} \right.
\end{equation}
\end{remark}

\subsection{Transforming to self-adjoint and Schr\"odinger form}
We now transform \eqref{eq:transformedHeun} to self-adjoint form, by the s-homotopic transformation
\begin{equation}\label{eq:u-def}
u(x) = (x-r_-)^{-s}(x-r_+)^{-2iM\omega} e^{-i\omega x} \tilde g(x)
\end{equation}
Then $u$ satisfies the equation
\begin{equation}\label{eq:u}
\partial_x \Delta \partial_x u(x) +  \tilde V_0(x) u(x) = 0
\end{equation}
where
\begin{equation}\label{eq:tildeV-def}
\tilde V_0(x)= - \Lambda + \Delta \omega^2 \left ( \frac{x+r_-}{x-r_+} \right )^2 - 4i\omega (x-r_-) \ximinus - \frac{x-r_+}{x-r_-} s^2
\end{equation}
\newcommand{\xs}{x_*}
Let $\xs$ be the tortoise coordinate corresponding to $x$,
$$
d\xs = \frac{x^2+a^2}{\Delta} dx
$$
Then, writing ${}' = d/d\xs$, and defining
\begin{align}
U ={}& (x^2+a^2)^{1/2} u , \label{eq:Udef}
\end{align}
we have the Schr\"odinger form of the transformed equation,
\begin{equation}\label{eq:schrod-transf}
U'' + \tilde V U = 0,
\end{equation}
where now
\begin{equation}\label{eq:Vtransf}
\tilde V = \frac{\Delta}{(x^2+a^2)^2} \tilde V_0 - V_{\text{resc}}
\end{equation}
with $V_{\text{resc}}$ as in \eqref{eq:Vresc}.

\begin{remark} \label{rem:Wronsk}
For $\omega \in \Reals$, the potential $\tilde V$ given by \eqref{eq:Vtransf} is real. Hence if $U$ is a solution to \eqref{eq:schrod-transf} the Wronskian $W[U, \bar U] = U' \bar U - \bar U' U$ is conserved, $W' = 0$, where $\bar U$ denotes the complex conjugate of $U$.
\end{remark}
We have
$$
\tilde V \big{|}_{x=r_+} = \frac{(r_+-r_-)^2}{r_+^2} \omega^2, \quad \lim_{x\to\infty} \tilde V(x) = \omega^2
$$

\begin{remark} \label{rem:yakov}
In the paper \citep{2015AnHP...16..289S} where the problem of mode stability on the real axis is considered for the case $s=0$, the integral transform \eqref{eq:inttrans} is applied with the contour $\Cont$ consisting of the real half-line starting at $r_+$, $\Cont=[r_+,\infty)$, to define $\tilde g(x)$ for $x$ with positive imaginary part. The function $\tilde g(x)$ is then extended to a domain including the real half line $[r_+,\infty]$ in the $x$-plane, and this extension is used to yield a solution to the Schr\"odinger type equation \eqref{eq:schrod-transf} (denoted $\tilde u$ in \citep{2015AnHP...16..289S}).

In this context it is important to emphasize that with $\Cont=[r_+,\infty)$, for real $x$ the  integral transform defining $\tilde g(x)$ does not converge absolutely, the integral form of $\partial_x \tilde g(x)$ obtained by differentiating under the integral sign is divergent, and the boundary condition \eqref{eq:bc} fails to be satisfied. In particular, the representation \eqref{eq:inttrans} of $\tilde g(x)$ fails to be valid for $x \in [r_+,\infty)$ and hence the solution $\tilde g(x)$ to \eqref{eq:transformedHeun} constructed by this extension procedure is \emph{different} from the function defined by the integral \eqref{eq:inttrans}.  
%See Appendix \ref{sec:yakov} for further discussion.
\end{remark}

\section{Limits} \label{sec:limits}
Recall \cite[Theorem C.2]{montgomery:vaughan} that for
$\Re \zeta > 0$,
\begin{equation}\label{eq:euler}
\int_0^{\infty} e^{-x} x^{\zeta - 1} dx = \Gamma(\zeta) , \quad \text{(Euler's integral).}
\end{equation}
where $\Gamma(\zeta)$ is the Gamma function. The Gamma function extends to a meromorphic function on the complex plane with simple poles at the non-positive integers.
Let $\rho > 0$ and let
$\Hankel$ be the contour in the complex $z$-plane which consists of the half line from  $-\infty-i\rho$ to $-i\rho$, the semi-circle of radius $\rho$ connecting $-i\rho$ with $i\rho$ and the half line from $i\rho$ to $i\rho - \infty$, see figure \ref{fig:Hankel}.
\begin{figure}[!h]
\centering

\raisebox{-0.5\height}{\includegraphics{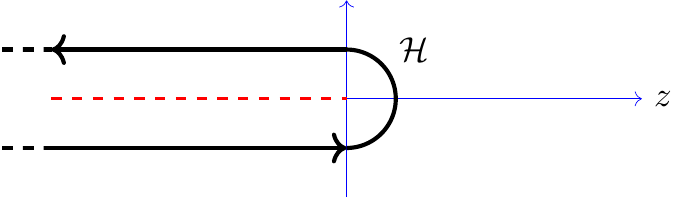}}

\caption{Hankel's contour} \label{fig:Hankel}
\end{figure}
Then for $\zeta \in \Co$ we have \cite[Theorem C.3]{montgomery:vaughan},
\begin{equation}\label{eq:hankel}
\frac{1}{2\pi i}\int_{\Hankel} e^z z^{-\zeta} dz = \frac{1}{\Gamma(\zeta)}, \quad \text{(Hankel's integral).}
\end{equation}
Among the many relations known for the Gamma function, we also recall the product formula \cite[Eq. (C.6)]{montgomery:vaughan}
$$
\Gamma(\zeta)\Gamma(1-\zeta) = \frac{\pi}{\sin(\pi\zeta)}
$$
We now consider the integral transform \eqref{eq:inttrans} for a contour $\Cont$. We restrict to the case
$$
\omega > 0 .
$$
and to contours $\Cont$ such that the integral \eqref{eq:inttrans} converges, and the boundary condition \eqref{eq:bc} is satisfied.

Let
$$
\nu = 2\omega(x-r_-)/(r_+-r_-),
$$
and define
\begin{equation}\label{eq:h-def}
h(r) = (r-r_-)^{\ximinus-s/2}(r-r_+)^{-\xiplus+s/2} e^{-i\omega r} R(r)
\end{equation}
We note that in view of \eqref{eq:int-asymp}, we have
\begin{equation}\label{eq:h-asymp}
h(r) \sim \left\{
\begin{array}{cc} 1 & \text{ for $r \to r_+$} \\
r^{-2\xiplus -s-1} & \text{ for $r \to \infty$}
\end{array} \right.
\end{equation}
We have the following corollary to Lemma \ref{lem:R-asympt}.
\begin{cor}\label{cor:h-est}
Let $h$ be given by \eqref{eq:h-def}.
Then, $h$ is analytic on the complex plane except at the singular points $r_-$, $r_+$ of the radial Teukolsky equation, 
where $h$ may have branch points. Further, it holds that
\begin{enumerate}
\item
In the non-resonant case, $\xi \ne 0$ or $s$ not a positive integer, $h(r)$ is analytic at $r_+$. 
\item
Let $S$ be the circular sector defined in Lemma \ref{lem:R-asympt}. Then
$$
h(r) = O(r^{-s-1})
$$
holds in $S$.
\end{enumerate}
\end{cor}
\begin{remark}
By 
Lemma \ref{lem:R-asympt}, $h$ is analytic for $|r - r_+| < |r_+ - r_-|$ if $s \leq 0$. 
\end{remark}

A calculation shows 
\begin{equation}\label{eq:hlim}
\lim_{r\to r_+} |h(r)|^2 = |r_+-r_-|^{-2s} \lim_{r\to r_+} \Delta^s |R(r)|^2
\end{equation}
For a given $\Cont$, we define, after choosing a suitable branch of $h$ if necessary,
\begin{equation}\label{eq:I-def}
I(\nu,\alpha) = \int_{\Cont} e^{i\nu(r-r_+)}(r-r_+)^{\alpha} h(r) dr \\
\end{equation}
Then
$$
\tilde g(x) = e^{i \nu(x) (r_+-r_-)} I(\nu(x),2\xiplus-s)
$$
In the rest of this section, we shall calculate the limit $\lim_{\nu \to \infty} \nu^{\alpha+1} I(\nu,\alpha)$. This argument is closely related to the proof of Watson's Lemma, cf. \cite[\S 1.9]{MR1034956}, which can be used to derive an expansion at $\nu=\infty$ of this expression. 
\subsection{The case $s > 0$}
Let $\rho_0> 0$ be small. We choose $\Cont$ to be the rotated Hankel type contour in the complex $r$-plane which consists of the half line from 
$i\infty + r_+ - \rho_0$ to $r_+ - \rho_0$, the semicircle of radius $\rho_0$ connecting $r_+ - \rho_0$  with $r_+ + \rho_0$, and the half line from $r_+ + \rho_0$ to $i\infty + r_+ + \rho_0$, 
see figure \ref{fig:r_+-contour}. 
\begin{figure}[!hb]
\centering

\raisebox{-0.5\height}{\includegraphics{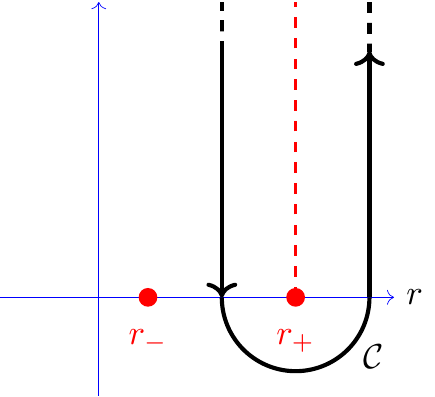}}

\caption{}%
\label{fig:r_+-contour}
\end{figure}
Using this contour in the definition of $\tilde g(x)$, we find that due to the exponential decay of the kernel $e^{i\nu (r-r_+)}$ for $\Im r \nearrow \infty$, the boundary condition \eqref{eq:bc} is satisfied with this choice of $\Cont$ and hence $\tilde g(x)$ is a solution to the transformed equation \eqref{eq:transformedHeun}.

We calculate
\begin{align}
I(\nu,\alpha) ={}& \int_{\Cont} e^{i\nu(r-r_+)}(r-r_+)^{\alpha} h(r) dr \nonumber \\
\intertext{substitute $\nu(r -r_+)= -it$}
={}& (-i)^{\alpha+1} \nu^{-\alpha-1} \int_{\Cont_{t}} e^{t} t^\alpha h(r_+ - \nu^{-1} i t) dt \label{eq:Inualpha-eval}
\end{align}
where $\Cont_t$ coincides with the Hankel contour with $\rho = \nu \rho_0$.

A limiting argument together with Hankel's integral formula \eqref{eq:hankel} now yields the following result.

\begin{lemma}\label{lem:lim-Hankel} Let $\Cont$ be the contour as in figure \ref{fig:r_+-contour}, and let $h$ satisfy the conclusions of Corollary \ref{cor:h-est}.
Then it holds that
$$
\lim_{\nu \to \infty} \nu^{\alpha+1} I(\nu,\alpha) = (-i)^{\alpha+1} \frac{2\pi i}{\Gamma(-\alpha)} h(r_+)
$$
\end{lemma}
\begin{proof}
If $r_+$ is a branch point for $h$, we choose a branch of $h$ by cutting the complex plan along the half line in the imaginary direction starting at $r_+$, see figure \ref{fig:r_+-contour}. In view of its definition, the integral $I(\nu,\alpha)$ is independent of $\rho_0$. Hence we can set $\rho_0 = \nu^{-1}\rho$ so that $\Cont_t$ coincides with the Hankel contour.
Starting from \eqref{eq:Inualpha-eval} we have,
\begin{align*}
\lim_{\nu \to \infty} \nu^{\alpha+1} I(\nu,\alpha) ={}& \lim_{\nu \to \infty}
(-i)^{\alpha+1}  \int_{\Cont_{t}} e^{t} t^\alpha h(r_+ - \nu^{-1} i t) dt \\
={}&
(-i)^{\alpha+1}  h(r_+) \int_{\Hankel} e^{t} t^\alpha   dt \\
={}&
(-i)^{\alpha+1} \frac{2\pi i}{\Gamma(-\alpha)} h(r_+)
\end{align*}
where in the last step we used \eqref{eq:hankel}.
\end{proof}
\begin{cor}\label{cor:Ulimit-Hankel} Assume that $R$ is a solution of the radial Teukolsky equation with $\omega>0$, and with no incoming radiation. Let $U$ be defined via \eqref{eq:Udef} and the integral transform \eqref{eq:inttrans} with the contour $\Cont$ as in figure \ref{fig:r_+-contour}. Then $U$ solves \eqref{eq:schrod-transf} and satisfies
\begin{equation}\label{eq:Ulimit-Hankel}
\lim_{x\to \infty} |U(x)| = C \lim_{r\to r_+} \Delta^{s/2} |R(r)|
\end{equation}
where
\begin{equation}\label{eq:Cdef-Hankel}
C = \left ( \frac{r_+ - r_-}{2\omega} \right )^{-s+1} \frac{2\pi}{|\Gamma(s-2\xiplus)|}|r_+-r_-|^{-s}
\end{equation}
and the limit on the left hand side of equation \eqref{eq:Ulimit-Hankel} is taken along the real axis.
\end{cor}
\begin{proof}
Due to the exponential decay of the kernel as $\Im r \to \infty$, the integral \eqref{eq:inttrans} converges, and the boundary condition \eqref{eq:bc} is satisfied. Therefore $\tilde g$ solves \eqref{eq:transformedHeun} and $U$ solves \eqref{eq:schrod-transf}. In order to evaluate the limit of $U$ at $x=\infty$, we use Lemma \ref{lem:lim-Hankel}.
We have 
\begin{align*}
\lim_{x\to \infty}|U(x)| ={}& \lim_{x\to \infty} |x^{-s+1} \tilde g(x)| \\
={}& \left ( \frac{r_+ - r_-}{2\omega} \right )^{-s+1} \lim_{\nu \to \infty} \nu^{-s+1} |I(\nu,2\xiplus - s)| \\
={}&  \left ( \frac{r_+ - r_-}{2\omega} \right )^{-s+1}
 \frac{2\pi}{|\Gamma(s-2\xiplus)|} |r_+ - r_-|^{-s} \lim_{r \to r_+} \KDelta^{s/2} |R(r)|
 \end{align*}
\end{proof}
\begin{remark} With the above choice of contour, Corollary \ref{cor:Ulimit-Hankel} is valid for arbitrary $s$.
For $s > 0$, $C$ defined by \eqref{eq:Cdef-Hankel} is bounded and non-zero. However, if $\xiplus =0$, then with $s$ a non-positive integer, the constant $C$ will vanish. Therefore, we shall in the next subsection consider a different contour which is more suitable for the case $s \leq 0$.
\end{remark}
\subsection{The case $s\leq 0$}
We consider the contour $\Cont$ consisting of the half line connecting $r_+$ with $i\infty$, i.e.
\begin{equation}\label{eq:r++it}
\Cont = \{ r \in [r_+,i \infty) \},
\end{equation}
see figure \ref{fig:r++it}.
\begin{figure}[!h]
\centering

\raisebox{-0.5\height}{\includegraphics{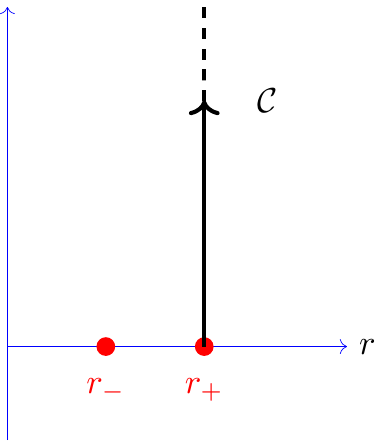}}

\caption{} \label{fig:r++it}
\end{figure}
Due to the exponential decay of the kernel $e^{i\nu(r-r_+)}$ as $\Im r \to \infty$, the boundary condition at $i\infty$ is automatically satisfied.
Further, due to $g(r) \sim 1$ at $r=r_+$, the boundary condition at $r=r_+$ is satisfied. 
Thus, with this choice of contour we have that $\tilde g(z)$ is a solution to the transformed equation \eqref{eq:transformedHeun}.

Starting with $I(\nu,\alpha)$ defined by \eqref{eq:I-def}, we calculate
\begin{align*}
I(\nu,\alpha) ={}& \int_{\Cont} e^{i\nu(r-r_+)}(r-r_+)^{\alpha} h(r) dr \\
\intertext{substitute $\nu(r -r_+)= it$}
={}& i^{\alpha+1} \nu^{-\alpha-1} \int_0^\infty e^{-t} t^\alpha h(r_+ + \nu^{-1} i t) dt
\end{align*}
A limiting argument together with Euler's integral formula \eqref{eq:euler} now yields the following result which is the analog of Lemma \ref{lem:lim-Hankel}.
\begin{lemma}\label{lem:lim-Euler}
Assume that $h$ satisfies the conclusions of Corollary \ref{cor:h-est}.
Then it holds that
$$
\lim_{\nu \to \infty} \nu^{\alpha+1} I(\nu,\alpha) = i^{\alpha+1} \Gamma(\alpha+1) h(r_+)
$$
\end{lemma}
Using the definitions we obtain the following analog of Corollary \ref{cor:Ulimit-Hankel}
\begin{cor}\label{cor:Ulimit-Euler} Assume that $R$ is a solution of the radial Teukolsky equation with $\omega >0$ and with no incoming radiation. Assume $s \leq 0$ and let $U$ be given by \eqref{eq:Udef}, and defined via the integral transform \eqref{eq:inttrans} with the contour $C$ given by \eqref{eq:r++it}. Then $U$ solves \eqref{eq:schrod-transf}, and we have
\begin{equation}\label{eq:Ulimit-Euler}
\lim_{x\to \infty} |U(x)| = C \lim_{x\to r_+} \Delta^{s/2} |R(r)|
\end{equation}
where
\begin{equation}\label{eq:Cdef-Euler}
C = \left ( \frac{r_+ - r_-}{2\omega} \right )^{-s+1} |\Gamma(2\xiplus-s+1)| |r_+-r_-|^{-s}
\end{equation}
 \end{cor}
\begin{proof} The fact that $U$ solves \eqref{eq:schrod-transf} follows as in Corollary \ref{cor:Ulimit-Hankel} due to the exponential decay of the kernel as $\Im r \to \infty$. Lemma \ref{lem:lim-Euler} yields
\begin{align*}
\lim_{x\to \infty}|U(x)| ={}& \lim_{x\to \infty} |x^{-s+1} \tilde g(x)| \\
={}& \left ( \frac{r_+ - r_-}{2\omega} \right )^{-s+1} \lim_{\nu \to \infty} \nu^{-s+1} |I(\nu,2\xiplus - s)| \\
={}&  \left ( \frac{r_+ - r_-}{2\omega} \right )^{-s+1} |\Gamma(2\xiplus-s+1)| |r_+ - r_-|^{-s} \lim_{r \to r_+} (\KDelta^{s/2} |R(r)|)
 \end{align*}
\end{proof}
\begin{remark} For $s \leq 0$, $C$ defined by \eqref{eq:Cdef-Euler} is bounded and non-zero.
\end{remark}

\section{Mode stability on the real axis} \label{sec:modestab}
We are now ready to prove our main result.
\begin{thm}\label{thm:realmodestab} Assume that $\omega > 0$ and that $R$ is a solution of the radial Teukolsky equation with no incoming radiation. Then $R = 0$.
\end{thm}
\begin{proof} We first consider the case $s\leq 0$.  Let $U(x)$ be given by \eqref{eq:Udef} and constructed via the integral transform \eqref{eq:inttrans} as explained above. By 
corollary \ref{cor:Ulimit-Euler},
$U(x)$ solves
\eqref{eq:schrod-transf}. By Remark \ref{rem:Wronsk} the Wronskian $W[U, \bar U] = U'\bar U -\bar U' U $ is conserved, $W' =0$.
From the definition of $U$, we can write it as
$$
U(x) = Z(x)I(\nu(x),2\xiplus -s).
$$
where
\begin{equation}\label{eq:Zdef}
Z(x) = (x^2+a^2)^{1/2}
(x-r_-)^{-s}(x-r_+)^{-2iM\omega} e^{-i\omega x}
e^{i \nu(x) (r_+-r_-)}
\end{equation}
We have
\begin{equation}\label{eq:WUZI}
W[U,\bar U] = W[Z,\bar Z] I\bar I + W[I, \bar I]Z\bar Z
\end{equation}
From \eqref{eq:I-def} we have by differentiating under the integral sign,
\begin{equation}\label{eq:dIdx}
\frac{d}{dx}I(\nu,\alpha) = i \frac{d\nu}{dx} I(\nu,\alpha+1)
\end{equation}
where
\begin{equation}\label{eq:dnudx}
\frac{d\nu}{dx} = \frac{2\omega}{r_+-r_-}
\end{equation}
Write $Z' = dZ/dx_*$ as in section \ref{sec:inttrans}. For $x\rightarrow r_+$, we have
\begin{align*}
Z' ={}& \frac{\Delta}{x^2+a^2} \left ( \frac{-2iM\omega}{x-r_+} \right ) Z(x) + O(x-r_+) \\
={}& \frac{r_+-r_-}{r_+^2+a^2} ( -2iM\omega) Z(x) + O(x-r_+) \\
={}& -4iM\kappa_+ \omega Z(x) + O(x-r_+)
\end{align*}
This gives
$$
W[Z,\bar Z](r_+) = -16iM^2\kappa_+ \omega r_+(r_+-r_-)^{-2s}
$$
In view of the discussion above $I(\nu(x),\alpha)$ and $dI(\nu(x),\alpha)/dx$ have well defined limits at $x=r_+$, and hence
$$
\lim_{x\to r_+} \frac{dI(\nu(x),\alpha)}{dx_*} = 0
$$
This gives
$$
W[U,\bar U](r_+) = W[Z,\bar Z](r_+) I(\nu(r_+),2\xiplus-s) \bar I (\nu(r_+),2\xiplus-s)
$$
In particular,
\begin{equation}\label{eq:iW-sign}
iW[U,\bar U](r_+) > 0
\end{equation}
for $|a| < M$.
We now consider the limit $x \to \infty$. 
Equations \eqref{eq:Zdef} and \eqref{eq:dnudx} give for large $x$,
\begin{equation}\label{eq:Z'}
Z' = (i\omega + O(1/x)) Z(x)
\end{equation}
which yields 
$$
Z'(x) \bar Z(x) = (i\omega + O(1/x) )Z(x)\bar Z(x)
$$
From \eqref{eq:dIdx} and 
corollary \ref{cor:Ulimit-Euler},
we have that
$$
\frac{dI(\nu(x),\alpha)}{dx_*} = O(1/x) I(\nu(x),\alpha)
$$
for large $x$. This means that
\begin{align*}
\lim_{x\to \infty} W[U,\bar U](x) ={}& \lim_{x\to\infty} W[Z,\bar Z](x) |I(\nu(x),2\xiplus-s)|^2 \\
={}&  2i\omega \lim_{x\to \infty} |Z(x)|^2 |I(\nu(x),2\xiplus-s)|^2 \\
={}& 2i\omega \lim_{x\to \infty} |U(x)|^2
\end{align*}
Write $|U(\infty)|^2 = \lim_{x\to\infty}|U(x)|^2$ and $|U(r_+)|^2 = \lim_{x\to r_+} |U(x)|^2$.
The conservation property of the Wronskian  gives for $\omega > 0$ and $|a| < M$
\begin{align*}
0 ={}& i(W(\infty) - W(r_+)) \\
={}& -2\omega |U(\infty)|^2 - iW(r_+) \\
\end{align*}
In view of \eqref{eq:iW-sign} this can hold only if $|U(\infty)|^2 = 0$ and $|U(r_+)|^2 = 0$, which implies $U(x) \equiv 0$. Taking into account the definition of $U(x)$, this implies that $R_{+,1}(r_+) = 0$, and hence that $R(r) \equiv 0$. 

For the case $s =\sfrak > 0$ we shall present two alternative approaches. Let $R_\sfrak$ denote a solution to the radial Teukolsky equation with no incoming radiation. It is straightforward to check that the TSI relation \eqref{eq:TSI-s} 
yields a solution $R_{-\sfrak}$ with no incoming radiation. 
In order to demonstrate that $R_{\sfrak} = 0$ it suffices to show that the solution $R_{-\sfrak}$ defined by 
\begin{equation}\label{eq:noTSI}
R_{-\sfrak} = \Delta^{\sfrak/2} {\mathcal{D}^\dagger}^{2\sfrak}\Delta^{\sfrak/2} R_\sfrak 
\end{equation} 
is non-vanishing. This follows due to the asymptotic form of  $R_{\sfrak}$ given by \eqref{eq:Rout-asympt}, and the fact that $\mathcal{D}^\dagger$ is to leading order 
$$
\mathcal{D}^{\dagger} = \frac{d}{dr} + i\omega + O(1/r)
$$
Arguing as in the first part of the proof, we find that $R_{-\sfrak}$ must vanish, and hence also $R_{\sfrak}$\footnote{The inference we want to draw from equation  
%fix revtex problem!? 
%(\ref{eq:noTSI}) 
\protect{\eqref{eq:noTSI}} 
is that if $R_{-\sfrak}=0$ then $R_{\sfrak}=0$ must also hold.  The argument for this fails at algebraically special modes \citep{1984RSPSA.392....1C}.  However, this fact is not an obstacle to our inference since algebraically special modes do not have no incoming radiation in the sense of definition  \ref{def:noincoming} and occur in case $C_{\sfrak}$ vanishes, which does not happen for real frequencies.}. 

An second, alternate argument for the case $s>0$ can be given as follows. 
For the non-resonant case, with $s>0$ we can argue along exactly the same lines as in the first part of the proof, but with corollary \ref{cor:Ulimit-Hankel} playing the role of corollary \ref{cor:Ulimit-Euler}. Finally, for the resonant case, with $\xiplus=0$ and $s$ a positive integer, we can apply the scattering relation \eqref{eq:scattering}. In the resonant case, we have $k_+ = 0$  which yields that the transmission coefficient vanishes, $\mathcal{T}_\sfrak = 0$. Assuming no incoming radiation, this yields 
$$
Y_{\text{hole},\sfrak} =0
$$
We now show that for spins $\sfrak=1,2$, the TSI constant $C_{\sfrak}$ given in equation \eqref{eq:Csfrak} is non-vanishing in the case of resonant frequencies, $k_+ = 0$, where $k_+$ is given by \eqref{eq:kpm-def}.
If $k_+ = 0$, then $\omega=am/(2Mr_+)$. 
For $\sfrak=1$, this gives
$$|C_1|^2= (\Lambda + 1)^2 + 4 a m \omega -4 a^2 \omega^2=(\Lambda + 1)^2+(8Mr_+-4a^2)\omega^2>0.$$
For $\sfrak=2$, $\Im C_2 = 12 M \omega$, cf. \citet[Eq. (3.25)]{1974ApJ...193..443T}, see also \citet[p. 462-463]{MR1647491}. This shows that $|C_2|^2 > 0$. 
Hence, 
$R_\sfrak=0$, which completes the proof of Theorem \ref{thm:realmodestab}. 
\end{proof}
\begin{remark}\label{rem:inhomog}
The radial Teukolsky equation \eqref{eq:radialeq} has conserved Wronskian
$$
W_{\RadialOp}[R_1, R_2] = \Delta \partial_r R_1 R_2 - \Delta R_1 \partial_r R_2,
$$
i.e. $\partial_r W_{\RadialOp}[R_1, R_2] = 0$ if $R_1, R_2$ solve \eqref{eq:radialeq}.
Let $R_{hor}$ and $R_{out}$ be solutions of the radial Teukolsky equation %
which are ingoing at the horizon and outgoing at infinity, respectively. Theorem \ref{thm:realmodestab} implies that $W_{\RadialOp}[R_{hor},R_{out}]$ is non-vanishing.

Consider an inhomogenous version of the radialy Teukolsky equation,
\begin{equation}\label{eq:radial-inhomog}
\RadialOp R = F
\end{equation}
In view of the above, we can use the method of variation of parameter to find a particular solution to \eqref{eq:radial-inhomog},
$$
R(r) = \frac{1}{W_{\RadialOp}[R_{out},R_{hor}]} \left (R_{hor}(r) \int_{r}^\infty R_{out}(t) F(t) dt + R_{out}(r) \int_{r_+}^r R_{hor}(t) F(t) dt \right )
$$
Due to the regular dependence of $W_{\RadialOp}$ on $\omega$ this can in principle be used to estimate the solution of the inhomogenous Teukolsky equation. This fact is related to the so-called quantitative mode stability, cf. \cite{2015AnHP...16..289S}.
\end{remark}

\subsection*{Acknowledgements} A substantial part of this work was carried out during the trimestre on Mathematical Relativity at the Institute Henri Poincar\'e, Paris, during the fall of 2015. L.A, S.M. and C.P. thank the IHP for support and hospitality. L.A. was partially supported by the CNRS. The work of BFW was supported in part by NSF Grants PHY 1205906 and PHY 1314529 to the University of Florida.  Support from the CNRS through the IAP, where part of this work was carried out, is also acknowledged, along with support from the French state funds managed by the ANR within the Investissements d'Avenir programme under Grant No. ANR-11-IDEX-0004-02. We are grateful to Steffen Aksteiner, Pieter Blue and Dietrich Hafner for helpful discussions.

\newcommand{\arxivref}[1]{\href{http://www.arxiv.org/abs/#1}{{arXiv.org:#1}}}
\newcommand{\mnras}{Monthly Notices of the Royal Astronomical Society}
\newcommand{\prd}{Phys. Rev. D}
\newcommand{\apj}{Astrophysical J.}

%\bibliography{kerr}
%merlin.mbs aipauth4-1.bst 2010-07-25 4.21a (PWD, AO, DPC) hacked
%Control: key (0)
%Control: author (9) reversed initials
%Control: editor formatted (0) differently from author
%Control: production of article title (0) allowed
%Control: page (1) range
%Control: year (1) truncated
%Control: production of eprint (0) enabled
%

\end{document}